\RequirePackage{fix-cm}
\documentclass[smallextended]{svjour3}       
\smartqed  
\usepackage{graphicx}
\usepackage{indentfirst,color}
\usepackage{amssymb}
\usepackage{color}
\usepackage{amsmath}
\usepackage{ulem}
\usepackage{paralist}
\usepackage[justification=centering]{caption}
\def\F{\mathbb{F}}

\def\C{\mathcal{C}}
\def\Tr{\text{\rm Tr}}
\def\wt{\text{\rm wt}}
\def\ker{\text{\rm ker}}
\def\ord{\text{\rm ord}}
\begin{document}
\baselineskip  6mm

\title{Weight distributions of several families of 3-weight binary linear codes
\thanks{This research is supported by National Natural Science Foundation of China (No. 61602342, 11701001) and Anhui
Provincial Natural Science Foundation (No. 1908085MA02). }
}
\subtitle{}


\author{Fei Li \ \ Xiumei Li
}


\institute{Fei Li \at
              \email{cczxlf@163.com}           
           \and
Faculty of School of Statistics and Applied Mathematics,
Anhui University of Finance and Economics, Bengbu, {\rm 233041}, Anhui, P.R.China\\
Xiumei Li \at \email{lxiumei2013@qfnu.edu.cn};
School of Mathematical Sciences,
Qufu Normal University, Qufu, {\rm 273165}, Shandong, P.R.China\\}

\date{Received: date / Accepted: date}

\maketitle

\begin{abstract}
The linear codes with a few weights have been
applied widely in combinatorial designs, secret sharing,
association schemes, authentication codes and strongly regular graphs.
In this paper, we first correct an erroneous result about the exponential sum
$\sum_{x\in \mathbb{F}_{2^{e}}}\chi_1\left(ax^{2^{\alpha}+1}+bx\right).$
Then, using the above exponential sum, we construct several families of binary linear codes of $3$-weight and
determine their weight distributions. Moreover, Most of them can be used in secret sharing schemes.

\keywords{weight distribution \and binary linear code \and exponential sum}
\subclass{ 94B05 \and 11T71}
\end{abstract}

\section{Introduction}
\label{intro}
Let $q=2^{e}$ for a positive integer $e$.
Denote by $ \F_q $ the finite field with $ q $
elements and $\F_q^{*}$ the multiplicative group of $\F_q$. Let
$g$ be a fixed generator of $\F_{q}^{*}$.

An $[n,k,\delta]$ binary linear code $\mathcal{C}$ is a $k$-dimensional subspace of $\mathbb{F}_{2}^{n}$
with minimum (Hamming) distance $\delta$. For $i\in\{1,2,\cdots,n\}, A_{i}$ denotes the number of codewords in $\mathcal{C}$ with weight $i$.
The sequence $(1,A_{1},A_{2},\cdots,A_{n})$ is called the weight distribution of the code $\mathcal{C}$ and the
polynomial defined by
$$
1+A_1x+A_2x^{2}+\cdots+A_{n}x^{n},
$$
is called the weight enumerator of $\C$. A code $\C$ is called $t$-weight if  $|\{i:A_i\neq 0,1\leq i\leq n\}| = t $.

The weight distribution of linear codes is a significant and hot research topic in coding theory and much attention has been paid \cite{5CK12,9DL16,DLM11,19LF08,SY20,23YY17}.
It can give the minimum distance of the code, hence the error correcting capability.
It is well-known that the weight distributions of codes allow the computation of the error
probability of error detection and correction with respect to some algorithms \cite{16K11}.

Recently, the linear codes with a few weights have become a hot research topic, since they are
applied widely in combinatorial designs \cite{O17}, secret sharing \cite{21YD06},
association schemes \cite{4CG84}, authentication codes \cite{8DH07} and strongly regular graphs \cite{5CK86}. Many studies about them have been done, see \cite{JL19,16KY19,LC17,LL18,26TX17,LY17}.

Ding et al. \cite{5DJ15,6DD14} gave the generic construction of
linear codes from defining set. Let $D=\{d_1,d_2,\ldots, d_n\} \subseteq \F_{q}^{*}$ and $\Tr$ denote the
trace function from $\F_q$ to $\F_2$. A binary linear code $\mathcal{C}_D$ of length $n$ is defined by
$$
\mathcal{C}_D=\{(\Tr(xd_1), \Tr(xd_2),\ldots, \Tr(xd_{n})):x\in \F_{q}\},
$$
where $D$ is called the defining set of $\mathcal{C}_D$. Using the above method, many linear codes with good parameters can be obtained by choosing the defining set D properly, such as \cite{25DD15,9DL16,SY20,19TXF17,23YY17,24ZL16}.

Motivated by Ding's construction, Li et al. \cite{LBY19} generalized Ding's method. Recall that the ordinary inner product of vectors $\mathbf{u}=(u_{1},u_{2},\cdots,u_{s})$, $\mathbf{v}=(v_{1},v_{2},\cdots,v_{s}) \in \mathbb{F}_{q}^{s} $ is
$$
\mathbf{\mathbf{u}}\cdot \mathbf{v}=u_{1}v_{1}+u_{2}v_{2}+\cdots+u_{s}v_{s}.
$$
Let $ D= \{\mathbf{d}_{1},\mathbf{d}_{2},\cdots,\mathbf{d}_{n}\}$
be a subset of $\F_{q}^{s}\backslash\{\mathbf{0}\} $,
a binary linear code $\mathcal{C}_D$ of length $n$ is defined by
\begin{eqnarray*}
         \C_{D}=\{\left( \Tr(\mathbf{x}\cdot \mathbf{d}_1), \Tr(\mathbf{x}\cdot \mathbf{d}_2),\ldots, \Tr(\mathbf{x}\cdot \mathbf{d}_{n})\right):\mathbf{\mathbf{x}}\in \mathbb{F}_{q}^{s}\}.
\end{eqnarray*}
Here the set $ D $ is also called the defining set of $\C_D$. Using the above generalized method, some classes of linear codes of two-weight and three weight are obtained, see \cite{JL19,LY17,LBY19}.

Inspired by the idea of Li et al. \cite{LBY19}, we choose the following defining set
\begin{equation}\label{def-set} D=D_{(a,b)}=\Big\{(x,y)\in \F_{q}^{2}\setminus\{\mathbf{0}\}: \Tr(ax^{2^{h}+1}+by)=0\Big\} = \{\mathbf{d}_1,\mathbf{d}_2,\ldots,\mathbf{d}_n\},
\end{equation}
where $a\in \mathbb{F}_{q}^{\ast}, \ b\in \mathbb{F}_{q}$ and $h$ is a proper divisor of $e$. The corresponding binary linear code $\C_{D}$ is defined by
\begin{equation}\label{defcode1}
\C_{D}=\Big\{\left( \Tr(\mathbf{x}\cdot \mathbf{d}_1), \Tr(\mathbf{x}\cdot \mathbf{d}_2),\ldots, \Tr(\mathbf{x}\cdot \mathbf{d}_{n})\right):\mathbf{x}\in \mathbb{F}_{q}^{2}\Big\}.
\end{equation}
In this paper, we mainly use exponential sum to determine their parameters and weight distributions.

The rest of this paper is organized as follows. In section 2 we introduce some basic
background such as trace function, canonical additive character and results about an exponential sum, which are very useful to get our results.
Moreover, we correct a result of Theorem 5.3 \cite{26LL18} and its proof.
In section 3 we present the parameters of several classes
of three-weight linear codes. We also give some examples. Section 4 summarizes this paper.


\section{Preliminaries}
\label{Pre}

In this section, we present some basic background concerning trace function, canonical additive character(see \cite{16LN97}).
Some known results about the exponential sum $S_{\alpha}(a,b)$ are also recalled.

Let $\Tr_{t}$ be the trace function from $\F_{q}$ to its subfield $\F_{2^{t}}$, that is,
for each $x\in \F_{q}$,
$$
\Tr_{t}(x)=x+x^{2^{t}}+ \cdots +x^{2^{t(\frac{e}{t}-1)}}.
$$
Simply, denote by $\Tr$ the \textit{absolute trace function} $\mathrm{Tr_{1}}.$
The \textit{canonical additive character} $\chi_{1}$ over $\F_{q}$ is defined by
$
\chi_{1}(x)=\exp(\frac{2\pi i\Tr(x)}{2}) = (-1)^{\Tr(x)}
$, where $x\in \F_{q}$.

Let $\alpha$ be a positive integer and $d=\gcd(e,\alpha)$. An exponential sum $S_{\alpha}(a,b)$ is defined as follows:
$$S_{\alpha}(a,b)=\sum_{x\in \F_{q}}\chi_1\left(ax^{p^{\alpha}+1}+bx\right),$$ where $ a\in\F_{q}^{*}, \ b\in \F_q$ and $p = 2$.

The evaluation of $S_{\alpha}(a,b)$ in odd characteristic ($p$ odd) were determined explicitly in \cite{5CA80,5CO98,5C98},
and in characteristic $2$ were solved in \cite{5CA79,26LL18}. In the following sequel, we shall give some lemmas that are essential in proving our main
results.

\begin{lemma}\cite[Theorem 4.1]{26LL18}\label{lem:1}  When $e/d$ is odd, we have
$$\sum_{x\in \F_{q}}\chi_1\left(ax^{2^{\alpha}+1}\right)=0, $$
for each $a\in\F_{q}^{*}.$ \end{lemma}

\begin{lemma}\cite[Theorem 4.2]{26LL18}\label{lem:2} Let $b\in\mathbb{F}_{q}^{*} $ and suppose $e/d$ is odd. Then
$$S_{\alpha}(a,b)=S_{\alpha}\left(1,bc^{-1}\right), $$
where $ c\in \mathbb{F}_{q}^{*} $ is the unique element satisfying
$ c^{2^{\alpha}+1}=a. $ Further we have
$$
S_{\alpha}(1,b)=\left\{\begin{array}{ll}
 0, & \textrm{if\ } \mathrm{Tr}_{d}(b)\neq 1, \\
 \pm 2^{\frac{e+d}{2}}, & \textrm{if\ } \mathrm{Tr}_{d}(b)= 1.
 \end{array}
 \right.
$$
\end{lemma}

\begin{lemma}\cite[Theorem 5.2]{26LL18}\label{lem:3} Let $e/d$ be even so that $e=2m$ for some integer $m$. Then
$$
S_{\alpha}(a,0)=\left\{\begin{array}{ll}
(-1)^{\frac{m}{d}}2^{m}, & \textrm{if\ } a\neq g^{t(2^{d}+1)}  \textrm{\ for any integer\ } t, \\
-(-1)^{\frac{m}{d}}2^{m+d}, & \textrm{if\ } a= g^{t(2^{d}+1)}  \textrm{\ for some integer\ } t,
\end{array}
\right.
$$
where $g$ is a generator of $\mathbb{F}_{q}^{*}$. \end{lemma}

\begin{lemma}\cite[Theorem 5.3]{26LL18}\label{lem:4}
Let $b\in\mathbb{F}_{q}^{*} $ and suppose $e/d$ is even so that $e=2m $
for some integer $ m$. Let $ f(x)=a^{2^{\alpha}}x^{2^{2\alpha}}+ax \in \mathbb{F}_{q}[x]$. There are two cases.\
\begin{itemize}
\item[(i)] If $a\neq g^{t(2^{d}+1)}$ for any integer  $t$, then $f$ is a permutation polynomial of
$\mathbb{F}_{q}$. Let $x_{0}$ be the unique element satisfying $f(x)=b^{2^{^{\alpha}}}$. Then
$$S_{\alpha}(a,b)
=(-1)^{\frac{m}{h}}2^{m}\chi_1\left(ax_{0}^{2^{\alpha}+1}\right).$$
\item[(ii)] If $a= g^{t(2^{d}+1)}$ for some integer $t$, then $S_{\alpha}(a,b)=0$ unless the equation
$f(x)=b^{2^{^{\alpha}}} $ is solvable. If this equation is solvable, with solution $ x_{0}$ say, then
$$
S_{\alpha}(a,b)
=\left\{\begin{array}{ll}
-(-1)^{\frac{m}{d}}2^{m+d}\chi_1\left(ax_{0}^{2^{\alpha}+1}\right), & \textrm{if\ } \Tr_{d}(a)= 0, \\
(-1)^{\frac{m}{d}}2^{m}\chi_1\left(ax_{0}^{2^{\alpha}+1}\right),  & \textrm{if\ } \Tr_{d}(a)\neq 0.
\end{array}
\right.
$$
\end{itemize}\end{lemma}

Here, we must note that the second part of the results of Lemma \ref{lem:4}(ii) is not right.
In fact, the evaluation of $S_{\alpha}(a,b)$ is equal to $-(-1)^{\frac{m}{d}}2^{m+d}\chi_1\left(ax_{0}^{2^{\alpha}+1}\right)$ whether $\Tr_{d}(a)$ is equal to $0$ or not.
We will prove that in Lemma \ref{lem:7}.

For proving Lemma \ref{lem:7}, we need the following Lemma \ref{lem:5} and Lemma \ref{lem:6}.
\begin{lemma}\cite[Lemma 4.2]{5C98}\label{lem:5}
Denote by $\chi_1$ the canonical additive character of $\mathbb{F}_{q}$ with $q = p^{e}, p$ any prime.
Let $a\in \mathbb{F}_{q}$ be arbitrary and let $d$ be some integer dividing $e.$ Then
$$
\sum_{x\in\mathbb{F}_{p^{d}}}\chi_1\left(ax\right)=\left\{\begin{array}{ll}
 p^{d}, & \ \ \textrm{if\ } \Tr_{d}(a)= 0, \\
 0, & \ \ \textrm{otherwise.\ }
 \end{array}
 \right.
$$
\end{lemma}

\begin{lemma}\cite[Lemma 2.1]{26LL18}\label{lem:6}
Let $d=\gcd(e,\alpha).$ Then
$$
\gcd(2^{\alpha}+1,2^{e}-1)=\left\{\begin{array}{ll}
 1, & \ \ \textrm{if\ } \ \frac{e}{d} \ \ \textrm{is odd,\ } \\
 2^{d} + 1, & \ \ \textrm{if\ } \ \frac{e}{d}\ \ \textrm{is even.\ }
 \end{array}
 \right.
$$
\end{lemma}

The following lemma is in fact a simple correction of Lemma \ref{lem:4}(ii)(\cite[Theorem 5.3,(ii)]{26LL18}) and the idea of its proof comes from Carlitz \cite{5C98} and Coulter \cite{26LL18}.

\begin{lemma}\label{lem:7}
Let $b\in\F_{q}^{*} $ and suppose $e/d$ is even so that $e=2m $
for some integer $ m$. Let $ f(x)=a^{2^{\alpha}}x^{2^{2\alpha}}+ax \in \F_{q}[x]. $ If $a= g^{t(2^{d}+1)}$ for some integer $t$, then $S_{\alpha}(a,b)=0$ unless the equation
$f(x)=b^{2^{^{\alpha}}} $ is solvable. If this equation is solvable, with solution $ x_{0}$ say, then
$$
S_{\alpha}(a,b)=-(-1)^{\frac{m}{d}}2^{m+d}\chi_1\left(ax_{0}^{2^{\alpha}+1}\right).
$$
\end{lemma}

\begin{proof}
As the proof of Theorem 5.3 in \cite{26LL18}, we have
$$
S_{\alpha}(a,b)S_{\alpha}(a,0)=\sum_{x\in \F_{q}}\Big(\chi_1(ax^{2^{\alpha}+1}+bx)\sum_{y\in \F_{q}}\chi_1\big(y^{2^{\alpha}}(f(x)+b^{2^{\alpha}})\big)\Big).
$$

Note that, if $x \in \F_{q}$ is not a solution of $f(x)=b^{2^{\alpha}}$, then the inner sum $\sum_{y\in \F_{q}}\chi_1\big(y^{2^{\alpha}}(f(x)+b^{2^{\alpha}})\big) = 0$.
If $x \in \F_{q}$ is a solution of $f(x)=b^{2^{\alpha}}$, then the inner sum $\sum_{y\in \F_{q}}\chi_1\big(y^{2^{\alpha}}(f(x)+b^{2^{\alpha}})\big) = q$.

Thus, if the equation $f(x)=b^{2^{\alpha}}$ is not solvable, then $S_{\alpha}(a,b)=0$.
If the equation $f(x)=b^{2^{\alpha}}$ is solvable, then, by Theorem 3.1 in \cite{26LL18}, we know that $f(x)=b^{2^{\alpha}}$ has $2^{2d}$ solutions in $\F_{q}$.
Denote by $x_{0}$ a fixed solution of $f(x) = b^{2^{\alpha}} $ and $\beta$ a fixed nonzero solution of $f(x) =0$,
then, the solutions of $f(x)=b^{2^{\alpha}}$ can be given by $x = x_{0} + \beta c$ with $c\in\F_{2^{2d}}$.
In this case, $\F_q $ is the disjoint union of
$\{x\in\F_q| \textrm{$x$ is a solution of the equation $f(x)=b^{2^{\alpha}}$} \}$ and $\{x\in\F_q| \textrm{$x$ is not solution of the equation $f(x)=b^{2^{\alpha}}$} \}$, so, we have
\begin{align*}
&S_{\alpha}(a,b)S_{\alpha}(a,0)=q\sum_{c\in\F_{2^{2d}}}\chi_1\left(a(x_{0} + \beta c)^{2^{\alpha}+1}+b(x_{0} + \beta c)\right)  \\
&=q\sum_{c\in\F_{2^{2d}}}\chi_1\left(ax_{0}^{2^{\alpha}+1}+bx_{0} + a(\beta c)^{2^{\alpha}+1} + a\beta cx_{0}^{2^{\alpha}}+ax_{0}(\beta c)^{2^{\alpha}}+b\beta c)\right)  \\
&=q\sum_{c\in\F_{2^{2d}}}\chi_1\left(ax_{0}^{2^{\alpha}+1}+bx_{0} + a(\beta c)^{2^{\alpha}+1} + (\beta c)^{2^{\alpha}}(a^{2^{\alpha}}x_{0}^{2^{2\alpha}}+ax_{0}+b^{2^{\alpha}}))\right)  \\
&=q\sum_{c\in\F_{2^{2d}}}\chi_1\left(ax_{0}^{2^{\alpha}+1}+bx_{0})\right)\chi_1\left(a(\beta c)^{2^{\alpha}+1} )\right),
\end{align*}
where we use the fact that $a^{2^{\alpha}}x_{0}^{2^{2\alpha}}+ax_{0}+b^{2^{\alpha}} = 0$ and $\chi_1$ is a additive group homomorphism.

Note that $f(x_{0}) = b^{2^{\alpha}}$ and $\Tr(z) = \Tr(z^{2^l})$ for any $z\in \F_q$ and non-negative integer $l$, we get
\begin{align*}
\mathrm{Tr}\left(ax_{0}^{2^{\alpha}+1}+bx_{0}\right)
&=\mathrm{Tr}\left(a^{2^{\alpha}}x_{0}^{2^{\alpha}}x_{0}^{2^{2\alpha}}+b^{2^{\alpha}}x_{0}^{2^{\alpha}}\right)  \\
&=\mathrm{Tr}\left(x_{0}^{2^{\alpha}}(ax_{0}+b^{2^{\alpha}})+b^{2^{\alpha}}x_{0}^{2^{\alpha}}\right)  \\
&=\mathrm{Tr}\left(ax_{0}^{2^{\alpha}+1}\right),
\end{align*}
which follows that
$$
S_{\alpha}(a,b)S_{\alpha}(a,0)=q\chi_1\left(ax_{0}^{2^{\alpha}+1}\right)\sum_{c\in\mathbb{F}_{2^{2d}}}\chi_1\left(a\beta^{2^{\alpha}+1}c^{2^{\alpha}+1}\right).
$$

By Lemma \ref{lem:6}, we know that $\gcd(2^{\alpha}+1, 2^{2d}-1)=2^{d}+1$. So, the map $N: \F_{2^{2d}}^{\ast}\rightarrow \F_{2^{d}}^{\ast}$ defined by $N(x)=x^{2^{\alpha}+1}$
is a surjective homomorphism. In fact, let $g_1$ be a generator of $\F_{2^{2d}}^{\ast}$, then
$$\ord(g_1^{2^{\alpha}+1})=\frac{2^{2d}-1}{\gcd(2^{\alpha}+1, 2^{2d}-1)} = 2^d-1,$$
where $\ord(g_1^{2^{\alpha}+1})$ is the order of $g_1^{2^{\alpha}+1}$, and thus, $g_1^{2^{\alpha}+1}$ is a generator of $\F_{2^{d}}^{\ast}$.

Now, we get
\begin{align*}
\sum_{c\in\mathbb{F}_{2^{2d}}}\chi_1\left(a\beta^{2^{\alpha}+1}c^{2^{\alpha}+1}\right) &= 1 + \sum_{c\in\F_{2^{2d}}^*}\chi_1\left(a\beta^{2^{\alpha}+1}c^{2^{\alpha}+1}\right)\\
&=1+(2^{d}+1)\sum_{y\in\mathbb{F}_{2^{d}}^{\ast}}\chi_1\left(a\beta^{2^{\alpha}+1}y\right).
\end{align*}

Since $f(\beta)=0,$ we have $(a\beta^{2^{\alpha}+1})^{2^{\alpha}-1}=1.$ Combining with
$$\gcd(2^{\alpha}-1, 2^{e}-1)=2^{\gcd(e,\alpha)}-1=2^{d}-1,$$
we can conclude that $(a\beta^{2^{\alpha}+1})^{2^{d}-1}=1,$ which means $a\beta^{2^{\alpha}+1} \in \mathbb{F}_{2^{d}}.$
Hence $\mathrm{Tr}_{d}(a\beta^{2^{\alpha}+1})= 0$ (using $e/d$ even). By Lemma 5,
$$
\sum_{y\in\mathbb{F}_{2^{d}}^{\ast}}\chi_1\left(a\beta^{2^{\alpha}+1}y\right)=2^d-1,$$
which follows that $$\sum_{c\in\mathbb{F}_{2^{2d}}}\chi_1\left(a\beta^{2^{\alpha}+1}c^{2^{\alpha}+1}\right)=2^{2d}.
$$
So, we have
$
S_{\alpha}(a,b)S_{\alpha}(a,0)=q\chi_1\left(ax_{0}^{2^{\alpha}+1}\right)2^{2d}.
$
Divided by $S_{\alpha}(a,0), $ the claimed result follows. The proof is finished. \end{proof}

Two examples are given below to test the correctness of Lemma \ref{lem:7}.

\begin{example}
Let $(\mathbb{F}_{q},\alpha,a,b)=(\mathbb{F}_{2^{6}},1,g^{3},g^{3}+g^{33}).$ It is easy to see $f(1)=b^{2^{\alpha}}.$ By Magma, we have $\mathrm{Tr}_{1}(a)=1\neq 0$ and
$$S_{\alpha}(a,b)=\sum_{x\in \mathbb{F}_{q}}\chi_1\left(g^{3}x^{3}+(g^{3}+g^{33})x\right)=-16.$$
\end{example}

\begin{example}
Let $(\mathbb{F}_{q},\alpha,a,b)=(\mathbb{F}_{2^{6}},1,g^{9},g^{9}+g^{36}).$ It is easy to see $f(1)=b^{2^{\alpha}}.$ By Magma, we have $\mathrm{Tr}_{1}(a)=0$ and
$$S_{\alpha}(a,b)=\sum_{x\in \mathbb{F}_{q}}\chi_1\left(g^{9}x^{3}+(g^{9}+g^{36})x\right)=16.$$
\end{example}

\section{Weight distributions of binary linear codes $\C_D$}

In this section, we study the weight distribution of linear code $\C_D$ in \eqref{defcode1}.

Recall that $h$ is a proper divisor of $e$ and $D=D_{(a,b)}$ in \eqref{def-set} with $a\in \F_{q}^{\ast}, \ b\in \F_{q}$. Write $S(a,b) = S_h(a,b)$ for short.
Now we first determine the length of the codes $\C_D$.

\begin{lemma}\label{lem:8} Let $a\in \F_{q}^{\ast}, \ b\in \F_{q}$.  Then, $$
n=|D|=\left\{\begin{array}{ll}
\frac{1}{2}q^{2}+\frac{1}{2}qS(a,0)-1, & \textrm{if\ } \ b=0, \\
\frac{1}{2}q^{2}-1, & \textrm{if\ } \ b\neq0.
\end{array}
\right. $$

\end{lemma}

\begin{proof} By definition, we have
\begin{align*}
|D_{(a,b)}|
&=\frac{1}{2}\Big(\sum_{x,y\in \F_{q}}\sum_{z\in \F_{2}}(-1)^{\Tr(z(ax^{2^{h}+1}+by))}\Big) - 1  \\
&=\frac{1}{2}\sum_{x,y\in \F_{q}}\Big(1+(-1)^{\Tr(ax^{2^{h}+1}+by)}\Big) - 1  \\
&=\frac{1}{2}q^{2}+\frac{1}{2}\sum_{x,y\in \F_{q}}(-1)^{\Tr(ax^{2^{h}+1}+by)} - 1  \\
&=\frac{1}{2}q^{2}+\frac{1}{2}\sum_{y\in \F_{q}}(-1)^{\Tr(by)}\sum_{x\in \F_{q}}(-1)^{\Tr(ax^{2^{h}+1})} - 1.
\end{align*}

If $b=0$, then
\begin{align*}
|D_{(a,0)}|
&=\frac{1}{2}q^{2}+\frac{1}{2}q\sum_{x\in \F_{q}}(-1)^{\Tr(ax^{2^{h}+1})} - 1  \\
&=\frac{1}{2}q^{2}+\frac{1}{2}qS(a,0) - 1 .
\end{align*}

If $b\neq0$, then $\sum_{y\in \F_{q}}(-1)^{\Tr(by)}=0$.  So,
$|D_{(a,b)}|=\frac{1}{2}q^{2}-1$.

 We complete the proof. \end{proof}

Let $c_{(u,v)}$ be the corresponding codeword in $\C_{D}$ with $(u,v)\in \F_q^2$, that is,
$$
c_{(u,v)}=\Big(\Tr(ux+vy)\Big)_{(x,y)\in D_{(a,b)}}.
$$
Obviously, $c_{(0,0)} = 0$ and $\wt(c_{(0,0)}) = 0$.
Next, we determine the Hamming weight $\wt(c_{(u,v)})$ with $(u,v)\neq \mathbf{0}$ in the following proposition.

\begin{proposition} \label{prop:1}
Let $(u,v)(\neq \mathbf{0})\in\F_{q}^{2}$. We have
\begin{enumerate}
\item[(i)] If $ b = 0 $, then
$$
\wt(c_{(u,v)})=\left\{\begin{array}{ll}
\frac{1}{4}q\Big(q+S(a,0)\Big), \  & \textrm{if\ } \ v\neq0, \\
\frac{1}{4}q\Big(q+S(a,0)-S(a,u)\Big), \  & \textrm{if\ } \ \ v=0.
\end{array}
\right.
$$
\item[(ii)] If $ b\neq 0$, then
$$
\wt(c_{(u,v)})=\left\{\begin{array}{ll}
\frac{1}{4}q^{2}, \  & \textrm{if\ } \  v=0 \ \textrm{or\ } \ v\neq b,\ v\neq 0,\\
\frac{1}{4}q\Big(q-S(a,u)\Big), \  & \textrm{if\ } \ v=b.
\end{array}
\right.
$$
\end{enumerate} \end{proposition}

\begin{proof} Put $N(u,v)=\{(x,y)\in \F_{q}^{2}: \Tr(ax^{2^{h}+1}+by)=0, \Tr(ux+vy) = 0\}$, then
\begin{align*}
|N(u,v)|
&=\frac{1}{4}\sum_{x,y\in \F_{q}}\Big(\sum_{z_{1}\in \F_{2}}(-1)^{\Tr(z_{1}(ax^{2^{h}+1}+by))}\sum_{z_{2}\in \F_{2}}(-1)^{\Tr(z_{2}(ux+vy))}\Big)  \\
&=\frac{1}{4}\sum_{x,y\in \F_{q}}\Big(\big(1+(-1)^{\Tr(ax^{2^{h}+1}+by)}\big)\big(1+(-1)^{\Tr(ux+vy)}\big)\Big)  \\
&=\frac{1}{4}\Big(q^{2}+\sum_{x,y\in \F_{q}}(-1)^{\Tr(ax^{2^{h}+1}+by)}+\sum_{x,y\in \F_{q}}(-1)^{\Tr(ax^{2^{h}+1}+by+ux+vy)}\Big),
\end{align*}
where $\sum_{x,y\in \F_{q}}(-1)^{\Tr(ux+vy)} =\sum_{x\in \F_{q}}(-1)^{\Tr(ux)}\sum_{y\in \F_{q}}(-1)^{\Tr(vy)} = 0 $.

If $b=0,$ then
\begin{align*}
|N(u,v)|
&=\frac{1}{4}\Big(q^{2}+\sum_{x,y\in \F_{q}}(-1)^{\Tr(ax^{2^{h}+1})}+\sum_{x,y\in \F_{q}}(-1)^{\Tr(ax^{2^{h}+1}+ux+vy)}\Big)  \\
&=\frac{1}{4}\Big(q^{2}+q\sum_{x\in \F_{q}}(-1)^{\Tr(ax^{2^{h}+1})}+\sum_{y\in \F_{q}}(-1)^{\Tr(vy)}\sum_{x\in \F_{q}}(-1)^{\Tr(ax^{2^{h}+1}+ux)}\Big)  \\
&=\frac{1}{4}\Big(q^{2}+qS(a,0)+S(a,u)\sum_{y\in \F_{q}}(-1)^{\Tr(vy)}\Big).
\end{align*}
So, we have
$$
|N(u,v)|=\left\{\begin{array}{ll}
\frac{1}{4}q\Big(q+S(a,0)\Big), \  & \textrm{if\ } \ v\neq0, \\
\frac{1}{4}q\Big(q+S(a,0)+S(a,u)\Big), \  & \textrm{if\ } \  v=0.
\end{array}
\right.
$$

Noting that $\wt(c_{(u,v)})=n-|N(u,v)|+1$.
By Lemma \ref{lem:8}, we have
$$
\wt(c_{(u,v)})=\left\{\begin{array}{ll}
\frac{1}{4}q\Big(q+S(a,0)\Big), \  & \textrm{if\ } \ v\neq0, \\
\frac{1}{4}q\Big(q+S(a,0)-S(a,u)\Big), \  & \textrm{if\ } \ \ v=0.
\end{array}
\right.
$$

If $b\neq0,$ then $$\sum_{x,y\in \F_{q}}(-1)^{\Tr(ax^{2^{h}+1}+by)} =\sum_{x\in \F_{q}}(-1)^{\Tr(ax^{2^{h}+1})}\sum_{y\in \F_{q}}(-1)^{\Tr(by)} = 0 ,$$
which follows that
\begin{align*}
|N(u,v)|
&=\frac{1}{4}q^{2}+\frac{1}{4}\sum_{x,y\in \F_{q}}(-1)^{\Tr(ax^{2^{h}+1}+by+ux+vy)}  \\
&=\frac{1}{4}q^{2}+\frac{1}{4}\sum_{y\in \F_{q}}(-1)^{\Tr((b+v)y)}\sum_{x\in \F_{q}}(-1)^{\Tr(ax^{2^h+1}+ux)}  \\
&=\frac{1}{4}q^{2}+\frac{1}{4}S(a,u)\sum_{y\in \F_{q}}(-1)^{\Tr((b+v)y)}.
\end{align*}

So, we have
$$
|N(u,v)|=\left\{\begin{array}{ll}
\frac{1}{4}q^{2}, \  & \textrm{if\ }  v=0, \\
\frac{1}{4}q\Big(q+S(a,u)\Big), \  & \textrm{if\ } \ v=b, \\
\frac{1}{4}q^{2}, \  & \textrm{if\ } \ v\neq b,\ v\neq0.
\end{array}
\right.
$$
By Lemma \ref{lem:8} again, we have
$$
\wt(c_{(u,v)})=\left\{\begin{array}{ll}
\frac{1}{4}q^{2}, \  & \textrm{if\ } \ v=0, \\
\frac{1}{4}q\Big(q-S(a,u)\Big), \  & \textrm{if\ } \ v=b, \\
\frac{1}{4}q^{2}, \  & \textrm{if\ } \ v\neq b,\ v\neq0.
\end{array}
\right.
$$

For the above, the proof is finished.
\end{proof}

\begin{remark}
By Proposition \ref{prop:1}, we know that, for $(u,v)(\neq \mathbf{0})\in\F_{q}^{2}$, we have $\wt(c_{(u,v)})>0$. So, the map: $\F_q^2\rightarrow\C_D$ defined by $(u,v)\mapsto c_{(u,v)} $
is an isomorphism as linear spaces over $\F_2$. Hence, the dimension of the codes $\C_D$ in \eqref{defcode1} is equal to $2e$.
\end{remark}

\begin{remark}\label{rem:2}
By the definition of $D$ in \eqref{def-set} and $\C_D$ in \eqref{defcode1}, we prove that the minimal distance of the dual code $C_{D}^{\perp}$ is at least $2$.
If not, then there exists a coordinate $i$ such that the $i$-th entry of all of the codewords of $\C_D$ is $0$,
that is, $\Tr(\mathbf{x}\cdot \mathbf{d}_{i})=0$ for all $\mathbf{x}\in \F_{q}^{2}$, where $\mathbf{d}_{i}\in D$.
Thus, by the properties of the trace function, we have $\mathbf{d}_{i}=0$. It contradicts with $\mathbf{d}_i \neq 0$.
\end{remark}

In the following sequel,  we will give the weight distribution of linear codes $\C_D$ in \eqref{defcode1} case by case on the term of the parity of $e/h$.

\begin{theorem}\label{thm:1}
Let $e/h$ be odd. The code $\C_{D_{(a,b)}}$ is a $[\frac{1}{2}q^{2}-1,2e,\delta]$ binary linear code
with the weight distribution in Table 1, where $\delta=\frac{1}{4}q(q-2^{\frac{e+h}{2}})$.
\begin{table}[ht]\label{tab:1}
\centering
 \caption{The weight distribution of the codes $\C_{D_{(a,b)}}$}
\begin{tabular}{|c|c|}
\hline
\textrm{Weight} $\omega$ \qquad& \textrm{Multiplicity} $A_\omega$   \\
\hline
0 \qquad&   1  \\
\hline
$\frac{1}{4}q^{2}$ \qquad&  $q^{2}-2^{e-h}-1$  \\
\hline
$\frac{1}{4}q(q-2^{\frac{e+h}{2}})$  \qquad& $2^{e-h-1}+2^{\frac{e-h-2}{2}}$  \\
\hline
$\frac{1}{4}q(q+2^{\frac{e+h}{2}})$  \qquad& $2^{e-h-1}-2^{\frac{e-h-2}{2}}$  \\
\hline
\end{tabular}
\end{table}
\end{theorem}

\begin{proof} We take $b =0$ for example. Assume that $(u,v)\neq \mathbf{0}$. By Lemma \ref{lem:1}, Lemma \ref{lem:2} and Proposition \ref{prop:1},
$\wt(c_{(u,v)})$ has only three values, that is,
\begin{equation*}
\left\{\begin{array}{ll}
\omega_{1}=\frac{1}{4}q^{2}, \  \\
\omega_{2}=\frac{1}{4}q(q-2^{\frac{e+h}{2}}), \  \\
\omega_{3}=\frac{1}{4}q(q+2^{\frac{e+h}{2}}).
\end{array}
\right.
\end{equation*}

Note that $A_{\omega_{i}}$ is the multiplicity of $\omega_{i}$. By Proposition~\ref{prop:1}, we have
\begin{align*}\label{eq:1}
  A_{\omega_{1}} &= \Big|\Big\{(u,v)\in\F_q^2|\wt(c_{(u,v)}) = \frac{1}{4}q^{2}\Big\}\Big|  \\
  &=\Big|\Big\{(u,v)\in\F_q^2|u\in\F_q,v\in\F_q^*\Big\}\Big| + \Big|\Big\{u\in\F_q^*|S(a,u)=0\Big\}\Big|\\
&=q(q-1) + \Big|\Big\{u\in\F_q^*|\Tr_h(uc^{-1})\neq 1\Big\}\Big|\\
&=q(q-1) + q-2^{e-h}-1,
\end{align*}
where $c\in\F_q$ satisfies $c^{2^h+1} = a$ and we use the fact that the map $\Tr_h$ is a surjective group homomorphism.

 Recall that the minimal distance of the dual code $\C_{D_{(a,0)}}^{\perp}$ is at least $2$. So, by the first two Pless Power Moment (\cite[P. 260]{12HP03} ),
we obtain the system of linear equations as follows:
\begin{equation*}
\left\{\begin{array}{ll}
A_{\omega_{1}}=q(q-1)+q-2^{e-h}-1 \  \\
A_{\omega_{2}}+A_{\omega_{3}}=2^{e-h} \  \\
\omega_{1}A_{\omega_{1}}+\omega_{2}A_{\omega_{2}}+\omega_{3}A_{\omega_{3}}=\frac{1}{2}q^{2}(\frac{1}{2}q^{2}-1).
\end{array}
\right.
\end{equation*}
Solving the above system, we get
\begin{equation*}
\left\{\begin{array}{ll}
A_{\omega_{1}}=q^{2}-2^{e-h}-1 \  \\
A_{\omega_{2}}=2^{e-h-1}+2^{\frac{e-h-2}{2}} \  \\
A_{\omega_{3}}=2^{e-h-1}-2^{\frac{e-h-2}{2}}.
\end{array}
\right.
\end{equation*}

For the case $b \neq 0$, by Lemma \ref{lem:2} and Proposition \ref{prop:1} and using the strategy used in the proof of case $b = 0$,
we can obtain required results. Thus,  We complete the proof.

\end{proof}

\begin{example}
Let $(q,h,a,b)=(2^{5},1,1,0)$. Then, the corresponding code $\C_{D_{(1,0)}}$ has parameters $[ 511,10,192]$ and weight enumerator
$1+10x^{192}+1007x^{256}+6x^{320}.$
\end{example}

\begin{example}
Let $(q,h,a,b)=(2^{5},1,1,1).$ Then, the corresponding code $\C_{D_{(1,1)}}$ has parameters $[ 511,10,192],$ weight enumerator
$1+10x^{192}+1007x^{256}+6x^{320}.$
\end{example}

Recall that $g$ is the fixed generator of $\F_{q}^*$. Denote by $\langle g^{2^{h}+1}\rangle$ the cyclic subgroup generated by $g^{2^{h}+1}$.
Now we will give the weight distribution of linear codes $\C_D$ in \eqref{defcode1} when $e/h$ is even.

\begin{theorem}\label{thm:2}
Let $b =0 $ and $e/h$ be even. Denote $m = e/2$. We have
\begin{enumerate}
\item[(i)] If $a\notin \langle g^{2^{h}+1}\rangle$, then
the code $\C_{D_{(a,0)}}$ is an $[n,2e]$ binary linear code
with the weight distribution in Table 2, where $n=\frac{1}{2}q\Big(q+(-1)^{\frac{m}{h}}2^{m}\Big)-1$.

\begin{table}[ht]\label{tab:2}
\centering
\caption{The weight distribution of the codes $\C_{D_{(a,0)}}$}
\begin{tabular}{|c|c|}
\hline
\textrm{Weight} $\omega$ \qquad& \textrm{Multiplicity} $A_\omega$   \\
\hline
0 \qquad&   1  \\
\hline
$\frac{1}{4}q(q+(-1)^{\frac{m}{h}}2^{m})$ \qquad&  $q(q-1)$  \\
\hline
$\frac{1}{4}q^{2}$  \qquad& $\frac{1}{2}q+(-1)^{\frac{m}{h}}2^{m-1}-1$  \\
\hline
$\frac{1}{4}q(q+(-1)^{\frac{m}{h}}2^{m+1})$  \qquad& $\frac{1}{2}q-(-1)^{\frac{m}{h}}2^{m-1}$  \\
\hline
\end{tabular}
\end{table}
\item[(ii)] If $a\in \langle g^{2^{h}+1}\rangle$, then the code $\C_{D_{(a,0)}}$ is an $[n,2e]$ binary linear code
with the weight distribution in Table 3, where $n=\frac{1}{2}q\Big(q-(-1)^{\frac{m}{h}}2^{m+h}\Big)-1.$

\begin{table}[ht]\label{tab:3}
\centering
\caption{The weight distribution of the codes of $\C_{D_{(a,0)}}$}
\begin{tabular}{|c|c|}
\hline
\textrm{Weight} $\omega$ \qquad& \textrm{Multiplicity} $A_\omega$   \\
\hline
0 \qquad&   1  \\
\hline
$\frac{1}{4}q(q-(-1)^{\frac{m}{h}}2^{m+h})$ \qquad&  $q^{2}-2^{e-2h}$  \\
\hline
$\frac{1}{4}q^{2}$  \qquad& $2^{e-2h-1}-(-1)^{\frac{m}{h}}2^{m-h-1}-1$  \\
\hline
$\frac{1}{4}q(q-(-1)^{\frac{m}{h}}2^{m+h+1})$  \qquad& $2^{e-2h-1}+(-1)^{\frac{m}{h}}2^{m-h-1}$  \\
\hline
\end{tabular}
\end{table}
\end{enumerate}
\end{theorem}

\begin{proof} Assume that $(u,v)\neq \mathbf{0}$. We discuss case by case on the term of $a\in \langle g^{2^{h}+1}\rangle$ or $a \notin\langle g^{2^{h}+1}\rangle$.

(i)  Suppose $a\notin \langle g^{2^{h}+1}\rangle$. As the proof of Theorem \ref{thm:1}, by Lemma \ref{lem:3}, Lemma \ref{lem:4} and Proposition \ref{prop:1},
$\wt(c_{(u,v)})$ has only three values, that is,
\begin{equation*}
\left\{\begin{array}{ll}
\omega_{1}=\frac{1}{4}q\Big(q+(-1)^{\frac{m}{h}}2^{m}\Big), \  \\
\omega_{2}=\frac{1}{4}q^{2}, \  \\
\omega_{3}=\frac{1}{4}q\Big(q+(-1)^{\frac{m}{h}}2^{m+1}\Big).
\end{array}
\right.
\end{equation*}
and we can obtain a system of linear equations as follows:
\begin{equation*}
\left\{\begin{array}{ll}
A_{\omega_{1}}=q(q-1) \  \\
A_{\omega_{2}}+A_{\omega_{3}}=q-1 \  \\
\omega_{1}A_{\omega_{1}}+\omega_{2}A_{\omega_{2}}+\omega_{3}A_{\omega_{3}}=\frac{1}{2}q^{2}n,
\end{array}
\right.
\end{equation*}
which concludes that
\begin{equation*}
\left\{\begin{array}{ll}
A_{\omega_{1}}=q(q-1) \  \\
A_{\omega_{2}}=\frac{1}{2}q+(-1)^{\frac{m}{h}}2^{m-1}-1 \  \\
A_{\omega_{3}}=\frac{1}{2}q-(-1)^{\frac{m}{h}}2^{m-1}.
\end{array}
\right.
\end{equation*}

(ii) Suppose $a\in \langle g^{2^{h}+1}\rangle$. By Lemma \ref{lem:3}, Lemma \ref{lem:7} and Proposition \ref{prop:1},
$\wt(c_{(u,v)})$ has only three values, that is,
$$
\left\{\begin{array}{ll}
\omega_{1}=\frac{1}{4}q\Big(q-(-1)^{\frac{m}{h}}2^{m+h}\Big), \  \\
\omega_{2}=\frac{1}{4}q^{2}, \  \\
\omega_{3}=\frac{1}{4}q\Big(q-(-1)^{\frac{m}{h}}2^{m+h+1}\Big).
\end{array}
\right.
$$

Recall that $ f(x)=a^{2^{h}}x^{2^{2h}}+ax \in \F_{q}[x]$ is a linear polynomial, so, it can be regarded as a linear map from $\F_q$ to itself.
Note that the number of solutions of $f(x)=0$ in $\F_{q}$ is $2^{2h}$. Hence, we have $\dim(\ker f)=2h$ and $\dim(f(\F_{q}))=e-2h$, which concludes that
\begin{align*}\label{eq:1}
  A_{\omega_{1}} &= \Big|\Big\{(u,v)\in\F_q^2|\wt(c_{(u,v)}) = \frac{1}{4}q\big(q-(-1)^{\frac{m}{h}}2^{m+h}\big)\Big\}\Big|  \\
  &=\Big|\Big\{(u,v)\in\F_q^2|u\in\F_q,v\in\F_q^*\Big\}\Big| + \Big|\Big\{u\in\F_q^*|S(a,u)=0\Big\}\Big|\\
&=q(q-1) + \Big|\Big\{u\in\F_q^*|f(x)=u^{2^h} \ \textrm{has no solution in}\  \F_q\Big\}\Big|\\
&=q(q-1) + q-2^{e-2h}.
\end{align*}

Thus,
we obtain the system of linear equations as follows:
\begin{equation*}
\left\{\begin{array}{ll}
A_{\omega_{1}}=q(q-1)+q-2^{e-2h} \  \\
A_{\omega_{2}}+A_{\omega_{3}}=2^{e-2h}-1 \  \\
\omega_{1}A_{\omega_{1}}+\omega_{2}A_{\omega_{2}}+\omega_{3}A_{\omega_{3}}=\frac{1}{2}q^{2}n,
\end{array}
\right.
\end{equation*}
which concludes that
\begin{equation*}
\left\{\begin{array}{ll}
A_{\omega_{1}}=q^{2}-2^{e-2h} \  \\
A_{\omega_{2}}=2^{e-2h-1}-(-1)^{\frac{m}{h}}2^{m-h-1}-1 \  \\
A_{\omega_{3}}=2^{e-2h-1}+(-1)^{\frac{m}{h}}2^{m-h-1}.
\end{array}
\right.
\end{equation*}
\end{proof}

\begin{example}
Let $(q,h,a,b)=(2^{6},1,g,0).$ Then the corresponding code $\C_{D_{(g,0)}}$ has parameters $[1791,12,768]$ and weight enumerator
$1+36x^{768}+4032x^{896}+27x^{1024}$.
\end{example}

\begin{example}
Let $(q,h,a,b)=(2^{6},1,1,0).$ Then, the corresponding code $\C_{D_{(1,0)}}$ has parameters $[ 2559,12,1024]$ and weight enumerator
$1+9x^{1024}+4080x^{1280}+6x^{1536}.$
\end{example}

\begin{example}
Let $(q,h,a,b)=(2^{6},1,g^{3},0).$ Then the corresponding code $\C_{D_{(g^{3},0)}}$ has parameters $[ 2559,12,1024]$ and weight enumerator
$1+9x^{1024}+4080x^{1280}+6x^{1536}.$
\end{example}

\begin{theorem}\label{thm:3}  Let $b\neq0$ and $e/h$ be even. Denote $m = e/2$. We have
\begin{enumerate}
\item[(i)] If $a\notin \langle g^{2^{h}+1}\rangle$, then
the code $\C_{D_{(a,b)}}$ is a $[\frac{1}{2}q^{2}-1,2e]$ binary linear code
with the weight distribution in Table 4.
\begin{table}[ht]\label{tab:4}
\centering
\caption{The weight distribution of the codes $\C_{D_{(a,b)}}$}
\begin{tabular}{|c|c|}
\hline
\textrm{Weight} $w$ \qquad& \textrm{Multiplicity} $A$   \\
\hline
0 \qquad&   1  \\
\hline
$\frac{1}{4}q^{2}$ \qquad&  $q^{2}-q-1$  \\
\hline
$\frac{1}{4}q(q-2^{m})$  \qquad& $\frac{1}{2}q+2^{m-1}$  \\
\hline
$\frac{1}{4}q(q+2^{m})$  \qquad& $\frac{1}{2}q-2^{m-1}$  \\
\hline
\end{tabular}
\end{table}
\item[(ii)] If $a\in \langle g^{2^{h}+1}\rangle$, then the code $\C_{D_{(a,b)}}$ is a $[\frac{1}{2}q^{2}-1,2e]$ binary linear code
with the weight distribution in Table 5.

\begin{table}[ht]\label{tab:5}
\centering
\caption{The weight distribution of the codes $\C_{D_{(a,b)}}$}
\begin{tabular}{|c|c|}
\hline
\textrm{Weight} $\omega$ \qquad& \textrm{Multiplicity} $A_\omega$   \\
\hline
0 \qquad&   1  \\
\hline
$\frac{1}{4}q^{2}$ \qquad&  $q^{2}-2^{e-2h}-1$  \\
\hline
$\frac{1}{4}q(q-2^{m+h})$  \qquad& $2^{e-2h-1}+2^{m-h-1}$  \\
\hline
$\frac{1}{4}q(q+2^{m+h})$  \qquad& $2^{e-2h-1}-2^{m-h-1}$  \\
\hline
\end{tabular}
\end{table}
\end{enumerate}
\end{theorem}

\begin{proof} Assume that $(u,v)\neq \mathbf{0}$. We also discuss case by case by $a\in\langle g^{2^{h}+1}\rangle$ or $a\notin\langle g^{2^{h}+1}\rangle$.

(i)  Suppose $a\notin \langle g^{2^{h}+1}\rangle$. As the proof of Theorem \ref{thm:2}, by Lemma \ref{lem:4} and Proposition \ref{prop:1},
$\wt(c_{(u,v)})$ has only three values, that is,
$$
\left\{\begin{array}{ll}
\omega_{1}=\frac{1}{4}q^{2}, \  \\
\omega_{2}=\frac{1}{4}q(q-2^{m}), \  \\
\omega_{3}=\frac{1}{4}q(q+2^{m}).
\end{array}
\right.
$$
Similarly,
we can obtain a system of linear equations as follows:
$$
\left\{\begin{array}{ll}
A_{\omega_{1}}=q(q-1)-1 \  \\
A_{\omega_{2}}+A_{\omega_{3}}=q \  \\
\omega_{1}A_{\omega_{1}}+\omega_{2}A_{\omega_{2}}+\omega_{3}A_{\omega_{3}}=\frac{1}{2}q^{2}(\frac{1}{2}q^{2}-1).
\end{array}
\right.
$$
Solving the system , we get
$$
\left\{\begin{array}{ll}
A_{\omega_{1}}=q^{2}-q-1 \  \\
A_{\omega_{2}}=\frac{1}{2}q+2^{m-1} \  \\
A_{\omega_{3}}=\frac{1}{2}q-2^{m-1}.
\end{array}
\right.
$$

(ii) Suppose $a\in \langle g^{2^{h}+1}\rangle$. By Lemma \ref{lem:7} and Proposition \ref{prop:1},
$\wt(c_{(u,v)})$ has only three values, that is,
$$
\left\{\begin{array}{ll}
\omega_{1}=\frac{1}{4}q^{2}, \  \\
\omega_{2}=\frac{1}{4}q(q-2^{m+h}), \  \\
\omega_{3}=\frac{1}{4}q(q+2^{m+h}).
\end{array}
\right.
$$
Similarly,
we can obtain a system of linear equations as follows:
$$
\left\{\begin{array}{ll}
A_{\omega_{1}}=q^{2}-1-2^{e-2h} \  \\
A_{\omega_{2}}+A_{\omega_{3}}=2^{e-2h} \  \\
\omega_{1}A_{\omega_{1}}+\omega_{2}A_{\omega_{2}}+\omega_{3}A_{\omega_{3}}=\frac{1}{2}q^{2}(\frac{1}{2}q^{2}-1).
\end{array}
\right.
$$
Solving the system, we get
$$
\left\{\begin{array}{ll}
A_{\omega_{1}}=q^{2}-1-2^{e-2h} \  \\
A_{\omega_{2}}=2^{e-2h-1}+2^{m-h-1} \  \\
A_{\omega_{3}}=2^{e-2h-1}-2^{m-h-1}.
\end{array}
\right.
$$
\end{proof}

\begin{example}
Let $(q,h,u,v)=(2^{6},1,g,1).$ Then, the corresponding code $\C_{D_{(g,1)}}$ has parameters $[ 2047,12,896]$ and weight enumerator
$1+36x^{896}+4031x^{1024}+28x^{1152}$.
\end{example}

\begin{example}
Let $(q,h,u,v)=(2^{6},1,1,1).$ Then, the corresponding code $\C_{D_{(1,1)}}$ has parameters $[ 2047,12,768]$ and weight enumerator
$1+10x^{768}+4079x^{1024}+6x^{1280}$.
\end{example}

\section{Concluding Remarks}
We correct a wrong result of an exponential sum and
construct several families of binary linear codes by the generalized method of defining set.
Using the exponential sum theory, we determine the weight distributions of the codes.
It is shown that the presented linear codes have three nonzero weights. By Magma,
some concrete examples have been given to verify the correctness of their corresponding results.

Let $w_{\min}$ and $w_{\max}$ denote the minimum and maximum nonzero weight of the linear code $\C_{D_{(a,b)}},$ respectively.
If the code $\C_{D_{(a,b)}}$ satisfies one of the following five conditions, then it can be easily checked that
$$
 \frac{w_{\min}}{w_{\max}}> \frac{1}{2}.
$$
\begin{enumerate}
\item In Theorems \ref{thm:1} and $e> h+1$.
\item In Theorems \ref{thm:2}(i) and $e> 3-(-1)^{\frac{m}{h}}$.
\item In Theorems \ref{thm:2}(ii) and $m> h+1+(-1)^{\frac{m}{h}}$.
\item In Theorems \ref{thm:3}(i) and and $m>1.$
\item In Theorems \ref{thm:3}(ii) and $m>h+1.$
\end{enumerate}

By the results in \cite{21YD06}, most of the codes $\C_{D_{(a,b)}}$ in the paper
are suitable for constructing secret sharing schemes with interesting properties.

\par  \vskip 0.5 cm


\begin{thebibliography}{}
\baselineskip  4mm
%
%


\bibitem{1AK17}
J. Ahn, D. Ka, C. Li, Complete weight enumerators of a class of linear codes. Designs, Codes and  Cryptogr. 83(1)(2017) 83-99.

\bibitem{4CG84}
Calderbank A.R., Goethals J.M., Three-weight codes and association schemes. Philips J. Res. 39(1984) 143-152.

\bibitem{5CK86}
Calderbank A.R., Kantor W.M., The geometry of two-weight codes. Bull. London Math. Soc. 18(1986) 97-122.

\bibitem{5CK12}
Choi, S.T., Kim, J.Y., No, J.S., Chung, H., Weight distribution of some cyclic codes. In: Proceedings
of the International Symposium on Information Theory (2012) 2911-2913

\bibitem{5CA79}
L. Carlitz, Explicit evaluation of certain exponential sums, Math. Scand. 44(1979) 5-16.

\bibitem{5CA80}
L. Carlitz, Evaluation of some exponential sums over a finite field, Math. Nachr. 96(1980) 319-339.

\bibitem{5CO98}
R.S. Coulter, Explicit evaluations of some Weil sums, Acta Arithmetica 83(1998) 241-251.

\bibitem{5C98}
R.S. Coulter, Further evaluations of Weil sums, Acta Arithmetica 86(1998) 217-226.

\bibitem{26LL18}
R. S. Coulter, On the evaluation of a class of Weil sums in characteristic 2, New Zealand J. Math. 28 (1999) 171--184.

\bibitem{5DJ15}
C. Ding, Linear codes from some 2-designs, IEEE Trans. Inf. Theory 61(6)(2015) 3265-3275.

\bibitem{25DD15}
K. Ding,  C. Ding, A class of two-weight and three-weight codes and their applications in secret sharing, IEEE Trans. Inf. Theory
61(11)(2015) 5835-5842.

\bibitem{6DD14}
K. Ding, C. Ding, Bianry linear codes with three weights, IEEE Commun. Lett.
18(11)(2014) 1879-1882.

\bibitem{8DH07}
C. Ding, T. Helleseth, T. Klove,  X Wang, A generic construction of Cartesian authen- tication codes.
IEEE Trans. Inf. Theory 53(6)(2007) 2229-2235.

\bibitem{9DL16}
C. Ding, C. Li, N. Li, Z. Zhou, Three-weight cyclic codes and their weight distributions. Discrete
Math. 339(2)(2016) 415-427.

\bibitem{DLM11}
C. Ding, Y. Liu, C. Ma, L. Zeng, The weight distributions of the duals of cyclic codes with two
zeros. IEEE Trans. Inf. Theory 57(12)(2011) 8000-8006.

\bibitem{12HP03}
W. C. Huffman and V. Pless,  Fundamentals of error-correcting codes,
Cambridge University Press, Cambridge, 2003.

\bibitem{JL19}
G. Jian, C. Lin and R. Feng, Two-weight and three-weight linear codes based on Weil sums, Finite Fields Appl. 57(2019) 92-107.

\bibitem{16K11}
T. Kl$\phi$ve, Codes for Error Detection. World Scientific, Hackensack, 2007.

\bibitem{16KY19}
X. Kong, S. Yang, Complete weight enumerators of a class of linear codes with two or three weights. Discrete Math. 342(11)(2019) 3166-3176.

\bibitem{LY17}
C. Li, Q. Yue, F. Fu, A construction of several classes of two-weight and three-weight linear codes.
Appl. Algebr. in Eng. Commun. 28(1)(2018)1-20.

\bibitem{LBY19}
C. Li, S. Bae, S. Yang: Some results on two-weight and three-weight linear codes, Adv. Math. Commun. 13(1) 195- 211. doi: 10.3934/amc.2019013.

\bibitem{16LN97}
 R. Lidl, H. Niederreiter,  Finite fields, Cambridge University Press, New York, 1997.

\bibitem{26TX17}
H. Liu, Q. Liao,  Several classes of linear codes with a few weights from defining sets over $F_p+uF_p$, Des. Codes Cryptogr. 87(1)(2017) 15-29.

\bibitem{LL18}
Y. W. Liu and Z. H. Liu, On some classes of codes with a few weights, Adv. Math. Commun. 12(2)(2018) 415-428.

\bibitem{LC17}
G. Luo, X. Cao, Five classes of optimal two-weight linear codes. Cryptogr. Commun. 10(5)(2017)1-17.

\bibitem{19LF08}
J. Luo, K. Feng,  On the weight distribution of two classes of cyclic codes. IEEE Trans. Inf. Theory
54(12)(2008) 5332-5344.

\bibitem{O17}
O. Olmez, A link between combinatorial designs and three-weight linear codes. Des. Codes Cryptogr. 86(9)(2017) 1-17.

\bibitem{SY20}
 Y. Song, J. Yang, Weight distribution of two classes of linear codes with a few weights. Sci. China Inf.
Sci., 2020, 63(7): 179103, https://doi.org/10.1007/s11432-018-9610-9.

\bibitem{19TXF17}
C. Tang, C. Xiang, K. Feng, Linear codes with few weights from inhomogeneous quadratic functions. Des. Codes Cryptogr. 83(3)(2017) 691-714.

\bibitem{21YD06}
J. Yuan, C. Ding, Secret sharing schemes from three classes of linear codes. IEEE Trans. Inf. Theory 52(1)(2006) 206-212.

\bibitem{23YY17}
S. Yang, Z. Yao, C. Zhao, The weight distributions of two classes of p-ary cyclic codes with few weights. Finite Fields Appl. 44(2017) 76-91.

\bibitem{24ZL16}
Z. Zhou, N. Li, C. Fan and T. Helleseth, Linear codes with two or three weights from quadratic bent functions, Des. Codes Cryptogr. 81(2)(2016) 283-295.
\end{thebibliography}


\end{document}